\documentclass{article}
\usepackage[T1]{fontenc}
\usepackage{slantsc}

\usepackage[american]{babel}
\usepackage[utf8x]{inputenc}
\usepackage{dsfont}
\usepackage{amsfonts}
\usepackage{placeins}

\usepackage[caption=false,font=footnotesize]{subfig}

\usepackage[disable]{todonotes}

\usepackage{etex}
\usepackage{etoolbox}
\usepackage{tikz}
\usetikzlibrary{calc,patterns,decorations.pathmorphing,decorations.markings,matrix,fit,decorations.pathreplacing}
\usetikzlibrary{decorations.markings}
\usetikzlibrary{arrows}
\tikzset{
  big arrow/.style={
    decoration={markings,mark=at position 1 with {\arrow[scale=1.5,#1]{>}}},
    postaction={decorate},
    shorten >=0.4pt},
  big arrow/.default=black}

\usepackage{xspace}
\usepackage{wrapfig}
\usepackage{tabularx}
\usepackage{multirow}
\usepackage{authblk}
\usepackage{amsthm}

\usepackage[pdfborder={0 0 0}]{hyperref} 

\newcommand{\ie}{i.e.,\xspace}
\newcommand{\eg}{e.g.,\xspace}
\newcommand{\etal}{et~al.\xspace}

\newcommand{\C}{\ensuremath{\mathbb{C}}}               
\newcommand{\R}{\ensuremath{\mathbb{R}}}               

\newcommand{\Rtwo}{\ensuremath{\R \rule{0.3mm}{0mm}^2}\xspace}      
\newcommand{\Rd}{\ensuremath{\R \rule{0.3mm}{0mm}^d}\xspace}        

\newcommand{\Obda}{W.\,l.\,o.\,g.\@\xspace}

\newcommand{\cgalPackage}[1]{{\emph{#1}\index{CGAL package@\cgal{} package!#1@\emph{#1}}}}

\newcommand{\iiDMinkowskiSumsPackage}{\cgalPackage{2D Minkowski Sums}}

\newcommand{\iiDArrangementsPackage}{\cgalPackage{2D Arrangements}}

\newcommand{\iiDTriangulationsPackage}{\cgalPackage{2D Triangulations}}

\newcommand{\ignore}[1]{}
\newcommand{\cgal}{{\sc Cgal}}
\newcommand{\cpationstyle}{\small\sf }

\usepackage{relsize}

\def\ifmonospace{\ifdim\fontdimen3\font=0pt }

\def\C++{%
\ifmonospace%
    C++%
\else%
    C\kern-.1667em\raise.30ex\hbox{\smaller{++}}%
\fi%
\spacefactor1000 }

\def\Csharp{%
\ifmonospace%
    C\#%
\else%
    C\kern-.1667em\raise.30ex\hbox{\smaller{\#}}%
\fi%
\spacefactor1000 }

\definecolor{polyALightColor}{rgb}{0.7,0.8,1}
\definecolor{polyBLightColor}{rgb}{1,0.85,0.6}
\definecolor{polyCLightColor}{rgb}{0.85,1,0.3}
\definecolor{polyDLightColor}{rgb}{1,1,0.75}

\definecolor{polyAColor}{rgb}{0.5,0.7,1}
\definecolor{polyBColor}{rgb}{1,0.72,0.225}
\definecolor{polyCColor}{rgb}{0.8,0.9,0.2}
\definecolor{polyDColor}{rgb}{0.8,0.8,0.6}

\definecolor{polyADarkColor}{rgb}{0.2,0.4,0.67}
\definecolor{polyBDarkColor}{rgb}{0.67,0.4,0.12}
\definecolor{polyCDarkColor}{rgb}{0.53,0.6,0.13}
\definecolor{polyDDarkColor}{rgb}{0.53,0.53,0.4}

\definecolor{polyABDarkColor}{rgb}{0.5,0.35,0.4}

\newtheorem{theorem}{Theorem}
\newtheorem{definition}{Definition}

\newtheorem{lemma}{Lemma}
\newtheorem{corollary}{Corollary}

\graphicspath{{./fig/}}

\title{Exact Minkowski Sums of Polygons With Holes}
\date{}
\author[1]{Alon Baram}
\author[1]{Efi Fogel}
\author[1]{Dan Halperin}
\author[2]{\\Michael Hemmer}
\author[2]{Sebastian Morr}
\affil[1]{School of Computer Science, Tel Aviv University, Israel}
\affil[2]{Dept. of Computer Science, TU Braunschweig, Germany}

\begin{document}

\maketitle

\setcounter{footnote}{1}
\footnotetext{Work by E.F.\ and D.H.\ has been supported in part by the Israel
Science Foundation (grant no. 1102/11), by the German-Israeli Foundation
(grant no. 1150-82.6/2011), and by the Hermann Minkowski--Minerva Center for
Geometry at Tel Aviv University.}
\refstepcounter{footnote}
\footnotetext{Work by S.M.\ and M.H.\ has been supported by Google Summer of
Code 2014.}

\begin{abstract}
We present an efficient algorithm that computes the Min\-kow\-ski sum of
two polygons, which may have holes.  The new algorithm is based on the
convolution approach. Its efficiency stems in part from a property for
Minkowski sums of polygons with holes, which in fact holds in any 
dimension: Given two polygons with holes, for each input polygon we
can fill up the holes that are relatively small compared to the other
polygon. Specifically, we can always fill up all the holes of at least
one polygon, transforming it into a simple polygon, and still obtain
exactly the same Minkowski sum. Obliterating holes in the input
summands speeds up the computation of Minkowski sums.

We introduce a robust implementation of the new algorithm, which
follows the Exact Geometric Computation paradigm and thus guarantees
exact results. We also present an empirical comparison of the
performance of Minkowski sum construction of various input examples, 
where we show that the
implementation of the new algorithm exhibits better performance than
several other implementations in many cases. In particular, we compared
the implementation of the new algorithm, an implementation of the
standard \emph{convolution} algorithm, and an implementation of the
\emph{decomposition} approach using various convex decomposition
methods,
including two new methods that handle polygons with
holes---one is based on vertical decomposition and the other is based
on triangulation.

The software has been developed as an extension of the
\iiDMinkowskiSumsPackage{} package of \cgal{}
(Computational Geometry Algorithms Library).
Additional information and supplementary material is available at
our project page \url{http://acg.cs.tau.ac.il/projects/rc}.
\end{abstract}

\section{Introduction}
\label{sec:introduction}
Let $P$ and $Q$ be two point sets in \Rd. The Minkowski sum of $P$ and
$Q$ is defined as $P \oplus Q = \left\{p+q\,|\,p \in P,q \in
Q\right\}$. In this paper we focus on the computation of Minkowski
sums of general polygons in the plane, that is, polygons that may have
holes. However, some of our results also apply to higher dimensions.
Minkowski sums are ubiquitous in many fields \cite{lm-cpq-04},
including robot motion planning \cite{l-rmp-91}, assembly planning
\cite{fh-papit-13}, and computer aided design \cite{ek-osm-99}.



\subsection{Terms, Definition, and Related Work}
\label{ssec:terms}
During the last four decades many algorithms to compute the Minkowski
sum of polygons or polyhedra were introduced. For exact
two-dimensional solutions see, \eg~\cite{fhw-caass-12}.  For
approximate solutions see, \eg~\cite{hmsvz-csaio-99} and
\cite{k-ccsou-93}.  For exact and approximate three-dimensional
solutions see, \eg~\cite{h-emsoe-09}, \cite{ml-gvamc-10},
\cite{l-smcmb-09}, and~\cite{vm-amsapm-06}.

Computing the Minkowski sum of two convex polygons $P$ and $Q$ is
rather easy. As $P \oplus Q$ is a convex polygon bounded by copies of
the edges of $P$ and $Q$ ordered according to their slope, the
Minkowski sum can be computed using an operation similar to merging
two sorted lists of numbers.  If the polygons are not convex, it is
possible to use one of the two following general approaches:

\paragraph{Decomposition}
Algorithms that follow the decomposition approach decompose $P$ and
$Q$ into two sets of convex sub-polygons.  Then, they compute the
pairwise sums using the simple procedure described above. Finally,
they compute the union of the pairwise sums. This approach was first
proposed by Lozano-P\'{e}rez~\cite{l-spcsa-83}. The performance of
this approach heavily depends on the method that computes the convex
decomposition of the input polygons. Flato~\etal~\cite{afg-pdecm-02}
described an implementation of the first exact and robust version of
the decomposition approach, which handles degeneracies. They also
tried different decomposition methods, but none of them handles
polygons with holes.

Ghosh~\cite{g-ucfms-93} introduced \emph{slope diagrams}---a data
structure that was used later on by some of us to construct Minkowski
sums of bounded convex polyhedra in 3D~\cite{fh-eecms-06}.  
Hachenberger~\cite{h-emsoe-09}
constructed Minkowski sums of general polyhedra in 3D.
Both implementations are based on the Computational Geometry
Algorithms Library (\cgal) and follow the Exact Geometric
Computation (EGC) paradigm. 

\def\scaleFactor{0.42}
\def\convA{%
  \begin{tikzpicture}[scale=\scaleFactor]
    \filldraw[color=polyADarkColor,fill=polyALightColor,thick]%
    (1,0)--(0,1)--(-1,0)--(0,-1)--cycle;
  \end{tikzpicture}}
\def\convB{%
  \begin{tikzpicture}[scale=\scaleFactor]
    \filldraw[color=polyBDarkColor,fill=polyBLightColor,thick]%
    (0,0)--(8,0)--(8,6)--(0,6)--cycle;
    \filldraw[color=polyBDarkColor,fill=white,thick](1,1)--(4,5)--(7,1)--cycle;
  \end{tikzpicture}}
\def\convC{%
  \begin{tikzpicture}[scale=\scaleFactor]
    \fill[polyCLightColor]%
    (-1,-0)--(-0,-1)--(8,-1)--(9,-0)--(9,6)--(8,7)--(-0,7)--(-1,6)--cycle;
    \fill[white](5.25,2)--(2.75,2)--(4,3.66667)--cycle;
    \draw[black,thick,big arrow](-1,-0)--(-0,-1);
    \draw[black,thick,big arrow](-0,-1)--(8,-1);
    \draw[black,thick,big arrow](8,-1)--(9,-0);
    \draw[black,thick,big arrow](9,-0)--(9,6);
    \draw[black,thick,big arrow](9,6)--(8,7);
    \draw[black,thick,big arrow](8,7)--(-0,7);
    \draw[black,thick,big arrow](-0,7)--(-1,6);
    \draw[black,thick,big arrow](-1,6)--(-1,-0);
    \draw[black,thick,big arrow](2,1)--(5,5);
    \draw[black,thick,big arrow,dotted](5,5)--(4,6);
    \draw[black,thick,big arrow,dotted](4,6)--(3,5);
    \draw[black,thick,big arrow](3,5)--(6,1);
    \draw[black,thick,big arrow,dotted](6,1)--(7,0);
    \draw[black,thick,big arrow,dotted](7,0)--(8,1);
    \draw[black,thick,big arrow,dotted](8,1)--(7,2);
    \draw[black,thick,big arrow](7,2)--(1,2);
    \draw[black,thick,big arrow,dotted](1,2)--(0,1);
    \draw[black,thick,big arrow,dotted](0,1)--(1,0);
    \draw[black,thick,big arrow,dotted](1,0)--(2,1);
    \node at (4,2.6) {$[0]$};
    \node at (9,7) {$[0]$};
    \node at (4,0.5) {$[1]$};
    \node at (1,1) {$[2]$};
    \node at (7,1) {$[2]$};
    \node at (4,4.9) {$[2]$};
  \end{tikzpicture}}

\begin{figure}[ttp]
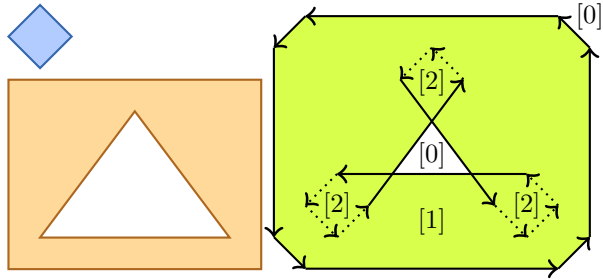

  \setlength{\tabcolsep}{2pt}
  \centering
  \begin{tabular}{lr}
    \convA&\multirow{2}*[20pt]{\convC}\\
    \convB&\\
  \end{tabular}
\caption[The convolutions of a convex and a concave polygons]{
  The convolution of a convex and a non-convex polygon;
  winding numbers are indicated in brackets;
  dotted edges are left out for the reduced convolution.
  }
  \label{fig:onecyc}
\end{figure}

\ignore{
\begin{figure}[ttp]
  \centering
  \begin{tabular}{p{\linewidth}}
    \begin{tabular*}{\linewidth}{l@{\extracolsep{\fill}}r}
      \convA & \convB
    \end{tabular*}\tabularnewline
    \begin{tabular*}{\linewidth}{l@{\extracolsep{\fill}}r}
      (a) The convex polygon. & (b) The concave polygon.\tabularnewline
    \end{tabular*}\tabularnewline
    \multicolumn{1}{c}{\convC}\tabularnewline
    \multicolumn{1}{c}{(c) The convolution and Minkowski sum.}
  \end{tabular}

  \caption[The convolutions of a convex and a concave polygons]{
    The convolution of a convex and a non-convex polygon;
    winding numbers are indicated in brackets;
    dotted edges are not part of the reduced convolution.
  }
  \label{fig:onecyc}
\end{figure}
}

\paragraph{Convolution}
Let $V_P = \left( p_0, \ldots, p_{m-1} \right)$ and $V_Q = \left( q_0,
\ldots, q_{n-1} \right)$ denote the vertices of the input polygons $P$
and $Q$, respectively. Assume that their boundaries wind in a
counterclockwise order around their interiors. The \emph{convolution}
of these two polygons, denoted $P * Q$, is a collection of line
segments of the form\footnote{Addition of vertex indices is carried
  out modulo $n$ for $P$ and modulo $m$ for $Q$.}
$[p_i + q_j, p_{i+1} + q_j]$, when the vector $\overrightarrow{p_i
  p_{i+1}}$ lies counterclockwise in between $\overrightarrow{q_{j-1}
  q_j}$ and $\overrightarrow{q_j q_{j+1}}$ and, symmetrically, of
segments of the form $[p_i + q_j, p_i + q_{j+1}]$, when the vector
$\overrightarrow{q_j q_{j+1}}$ lies counterclockwise in between
$\overrightarrow{p_{i-1} p_i}$ and $\overrightarrow{p_i p_{i+1}}$.

According to the \emph{Convolution Theorem} stated in $1983$ by
Guibas~\etal~\cite{grs-kfcg-83}, the convolution $P * Q$ of two
polygons $P$ and $Q$ is a superset of the boundary of the Minkowski
sum $P \oplus Q$. The segments of the convolution form a number of
closed (possibly self-intersecting) polygonal curves called
\emph{convolution cycles}. The set of points having a nonzero winding
number with respect to the convolution cycles comprise the Minkowski
sum $P \oplus Q$.\footnote{Informally speaking, the winding number of
  a point ${p \in \Rtwo}$ with respect to some planar curve $\gamma$
  is an integer number counting how many times does $\gamma$ wind in a
  counterclockwise orientation around $p$.} However, this theorem has
not been completely proven. In the introduction of the thesis of
Ramkumar~\cite{r-ttcta-98}, there are statements about the correctness
of the \emph{Convolution Theorem}. Some of these statements are still
given without proofs.

Wein~\cite{w-eecpm-06} implemented the standard convolution algorithms
for simple polygons. He computed the winding number for each face in
the arrangement induced by the convolution cycles and used it to
determine whether the face is part of the Minkowski sum or not; see
Figure~\ref{fig:onecyc}. Wein's implementation is available in
\cgal~\cite{cgal:w-rms2-15}, and as such, it follows the EGC paradigm.


Kaul~\etal~\cite{kos-cmsrp-91} observed that a segment $[p_i + q_j,
  p_{i+1} + q_j]$ (resp.\ $[p_i + q_j, p_i + q_{j+1}]$) cannot
possibly contribute to the boundary of the Minkowski sum if $q_j$
(resp.\ $p_j$) is a reflex vertex (see dotted edges in
Figure~\ref{fig:onecyc}).  The remaining subset of convolution
segments, the \emph{reduced convolution}, is still a superset of the
Minkowski sum boundary, but the idea of winding numbers can not be
applied any longer as there are no closed cycles anymore.  Instead,
Behar and Lien~\cite{bl-frmsu-11}, first identify faces in the
arrangement of the reduced convolution that may represent holes (based
on proper orientation of all boundary edges of the face).  Thereafter,
they check whether such a face is indeed a proper hole by selecting a
point $x$ inside the face and performing a collision detection of $P$
and $x\oplus-Q$.
Their implementation exhibits faster running time than Wein's implementation.
However, although it uses advanced multi-precision arithmetic,
it does not handle some degenerate cases correctly.
The method was also extended to three dimensions~\cite{l-smcmb-09}.

Milenkovic and Sacks~\cite{ms-mcms-07} defined the \emph{Monotonic
  Convolution}, which is another superset of the Minkowski sum
boundary. They show that this set defines cycles and induces winding
numbers, which are positive only in the interior of the Minkowski sum.

\ignore{
Convolution based methods typically outperform decomposition methods
in those cases where input polygons cannot be
easily decomposed, as
the number of segments in the convolution is
usually smaller than the number of segments that constitute the
boundaries of the sub-sums when using the decomposition approach.
}

\subsection{Our Results}
\label{ssec:results}
We present an efficient algorithm that computes the Minkowski sum of
two polygons, which may have holes. The new algorithm is a variant of
the algorithm proposed by Behar and Lien~\cite{bl-frmsu-11}, which
computes the reduced convolution set. In our new algorithm, the
initial set of filters proposed in~\cite{bl-frmsu-11} is enhanced by
the removal of complete holes in the input. This enhancement reduces
the size of the reduced convolution set even further. The enhancement
is backed up by a theorem, the proof of which is also presented; see
Section~\ref{sec:filtration}. Moreover, we show that at least one of
the input polygons can always be made simple (before applying the
convolution). These latter results are applicable to any dimension and are 
independent of the used approach. In addition, roughly speaking, we
show that every boundary cycle of the Minkowski sum is caused by
exactly one boundary cycle of each summand; see
Section~\ref{sec:filtration}. It implies that we can compute the
convolution of each pair of boundary cycles of the summands
separately.  This result is also applicable to any dimension and it is
independent of the used approach. However, applying it to the
decomposition approach requires the ability to handle unbounded
polygons.

We introduce an implementation of the new algorithm.
%
We also introduce implementations of two new convex decomposition
methods that handle polygons with holes as input---one is based on
vertical decomposition and the other is based on triangulation. These
two methods can be directly applied to compute the Minkowski sum of
polygons with holes via decomposition.  All our implementations are
robust and handle degenerate cases.
%

We present an empirical comparison of all the implementations above
and existing implementations; see Section~\ref{sec:experiments}. We
show that the implementation of our new algorithm, which computes the
reduced convolution set, exhibits better performance than all other
implementations in many cases.

\section{Filtering Out Holes}
\label{sec:filtration}
The fundamental observation of the convolution theorem is that only
points on the boundary of $P$ and $Q$ can contribute to the boundary
of $P\oplus Q$. 
Specifically, the union of the segments in the convolution $P*Q$, 
as a point set, is a super-set of the union of the segments 
of the boundary of $P\oplus Q$. 

The idea behind the reduced convolution method is to filter out
segments of $P*Q$ that can not possibly contribute to the boundary of
$P\oplus Q$ using a local criterion; see Section~\ref{ssec:terms}.  In
this section we introduce a global criterion. We show that if a hole
in one polygon is relatively small compared to the other polygon, the
hole is irrelevant for the computation of $P\oplus Q$; 
see Figure~\ref{fig:path_and_hole} for an illustration. 
Thus, we can ignore all segments in $P*Q$ that are induced by the hole when
computing $P\oplus Q$. It implies that the hole can be removed (that
is, filled up) before the main computation starts, regardless of 
the approach that one uses to compute the Minkowski sum.

\begin{figure}[htb]
\center
\includegraphics[width=0.7\columnwidth]{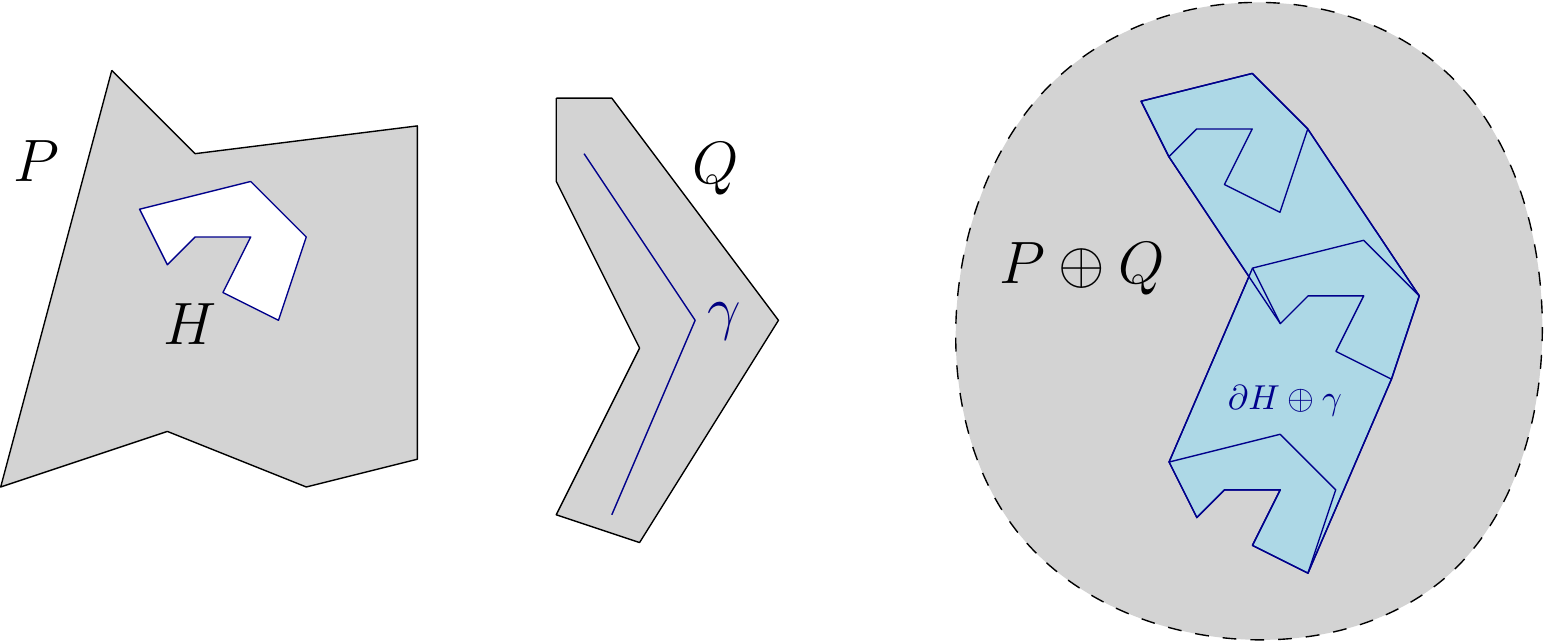}
\caption{\cpationstyle%
  A small hole $H$ is irrelevant for the computation of $P\oplus Q$
  as adding $\partial H$ and $\gamma \subset Q$ fills up any 
  potential hole in $P\oplus Q$ related to $H$.
\label{fig:path_and_hole}}
\end{figure}

\begin{definition}
  A hole $H$ of polygon $P$ leaves a trace in $P \oplus Q$, if there
  exists a point $r = p+q \in \partial(P \oplus Q)$, such that $p \in
  \partial H$ and $q \in \partial Q$. We say that $r$ is a trace of
  $H$.  Conversely, we say that a hole $H$ is irrelevant for the
  computation of $P \oplus Q$ if it does not leave a trace at all.
\end{definition}

\begin{lemma}\label{lemmaA}
  If $H$ leaves a trace in $P \oplus Q$ at a point $r$, then $r$ is on
  the boundary of a hole $\tilde H$ in $P \oplus Q$.
\end{lemma}

\begin{proof}
  Consider the point $r = p+q$, 
  which is on the boundary of $P \oplus Q$, 
  such that $p\in \partial H$ and $q \in \partial Q$. 
  Since the polygons are closed, for every neighborhood of $r$
  there exists a point $r' \not \in P \oplus Q$, see Figure \ref{fig:trace}. 
  Consequently, its corresponding point $p'=r'-q$, 
  which is in the neighborhood of $p$ must be in $H$. 
  Thus, $r'$ must be enclosed by $\partial H \oplus q$, 
  implying that $r'$ is inside a hole of $P \oplus Q$. 
%
\end{proof}

\begin{figure}[htb]
  \center
  \includegraphics[width=0.7\columnwidth]{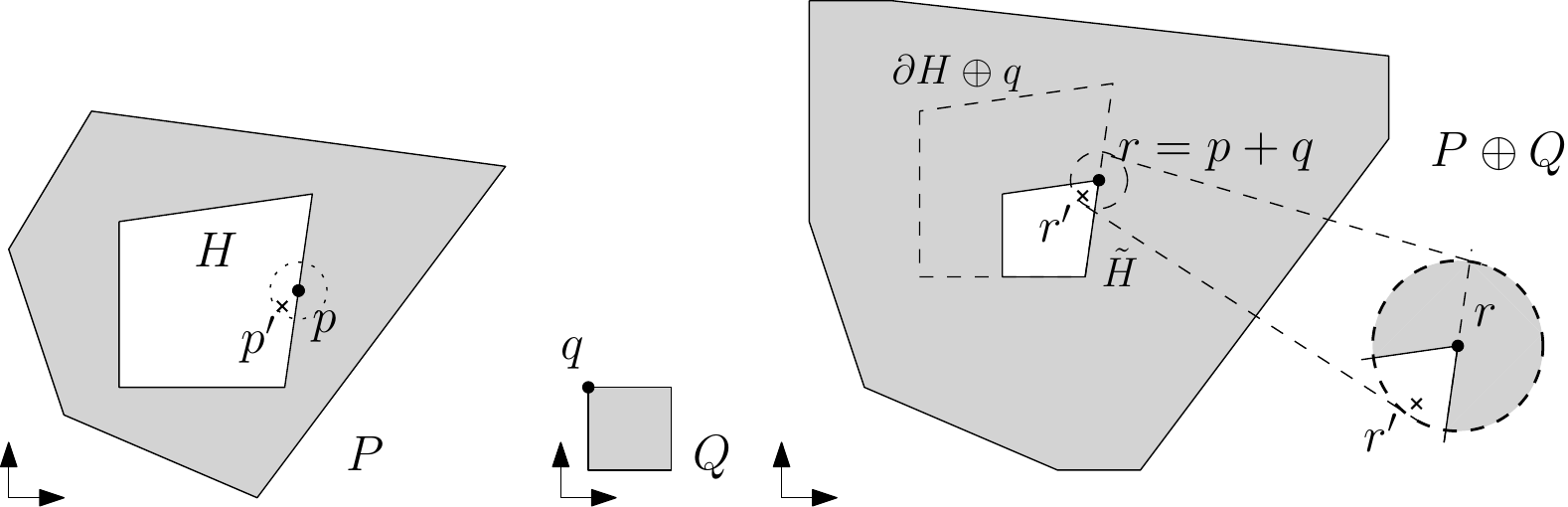}
  \caption{\cpationstyle%
    Hole $H$ leaves a trace in $P \oplus Q$ at point $r$, 
    which must be on the boundary of some hole $\tilde H$ in $P \oplus Q$; 
    see Lemma~\ref{lemmaA}.
    \label{fig:trace}}
\end{figure}

\begin{lemma}\label{lemmaB}
  Let $\tilde H$ be a hole in $P \oplus Q$ that contains a point $r
  = p + q \in \partial \tilde H$, such that $p \in \partial H$ and
  $q \in \partial Q$; that is, $r$ is a trace of $H$. Then $\forall z
  \in \tilde H$ and $\forall y \in Q$, it must hold that $z \in H
  \oplus y$. 
  In other words, $\tilde H \subseteq \bigcap_{\forall y \in Q} H \oplus y$.

\end{lemma}

\begin{proof}
  As in Lemma \ref{lemmaA}, there is a point $r'$ in the neighborhood of
  $r$, which is enclosed by $\partial H \oplus q$. Furthermore, there
  exists a continuous path from $r'\in\tilde H$ to any $z\in\tilde H$, 
  which means that every $z\in\tilde H$ is also enclosed by 
  $\partial H \oplus q$,  or in other words: $z \in H  \oplus q$.

   Now, assume for contradiction that there is a point $y_0 \in Q$, for
  which $z$ is not in $H \oplus y_0$. Consider the continuous
  path $\gamma$ that connects $q$ and $y_0$ within $Q$.  Observe that
  $z \in H \oplus q$ and $z \notin H \oplus y_0$ are equivalent to $z
  - q \in H$ and $z - y_0 \notin H$, respectively.  This means that $z
  - y_0$ is either in the unbounded face, or in some other hole in
  $P$.  Now observe that the path $z \oplus (-\gamma)$ connects $z -
  q$ and $z - y_0$.  
  Thus, since $\gamma$ is continuous, there must be a point $y_0' \in \gamma \subset
  Q$, for which $z - y_0' \in P$. 
  Hence, $z \in P \oplus y_0'$, which implies $z \in P \oplus Q$---a
  contradiction.
\end{proof}

\begin{corollary}\label{cor:one-trace}
  Let $\tilde H$ be a hole in $P \oplus Q$ with $r \in \partial \tilde H$
  being a trace of $H$. 
 Then $\forall s \in \partial \tilde H$ it holds that $s$ is a trace of $H$. 
\end{corollary}

\begin{proof}
  By Lemma~\ref{lemmaB} 
  $\tilde H \subseteq \bigcap_{\forall y \in Q} H \oplus y$
\end{proof}

\begin{theorem}\label{theorem1}
  Let $H$ be a closed hole in polygon $P$. 
  $H$ is irrelevant for the computation of $P\oplus Q$ 
  iff 
  there is a path contained in polygon $Q$ that does not fit 
  under any translation in~$-H$. 
\end{theorem}

\begin{proof}
  We first show that $H$ is irrelevant for the computation of $P\oplus Q$ 
  if there is a path $\gamma \subset Q$ 
  that does not fit under any translation in~$-H$.
  Assume for contradiction
  that $H$ leaves a trace in $P \oplus Q$; that is, there is an $r = p
  + q \in \partial (P \oplus Q)$, such that $p\in\partial H$ and $q
  \in \partial Q$. 
  By Lemma~\ref{lemmaA}, the point $r$ is on the boundary of 
  a hole $\tilde H$ in $P \oplus Q$. 
  By Lemma~\ref{lemmaB}, for any point $x \in
  \tilde H$ it must hold that $x \in H \oplus y$ $\forall y\in Q$.
  Specifically, it must hold $\forall y \in \gamma \subset Q$.  This
  is equivalent to $y \in (x \oplus -H)$ for all $y\in \gamma$,
  stating that $\gamma$ fits into $-H$ under some translation---a
  contradiction.

  \medskip 

  \noindent Conversely, if there is no path that does not fit into $-H$ then 
  all paths contained in $Q$ fit in $-H$. 
  Thus, also $Q$ itself fits in $-H$ under some translation $x$ with 
  $x \oplus Q \subseteq -H$.  
  In this case $x+q \in -H$ for all $q \in Q$, which is equivalent to $-x \in
  H \oplus q$ for all $q \in Q$. 
  This implies that $-x \notin P \oplus Q$, whereas 
  $-x \in (P\cup H) \oplus Q$, that is, H is relevant for $P \oplus Q$.
\end{proof}

\begin{corollary}\label{cor:hole_filter}
  If the closed axis-aligned bounding box $B_Q$ of $Q$ does not fit
  under any translation in the open axis-aligned bounding box
  $\mathring B_H$ of a hole $H$ in $P$, then $H$ does not have a trace
  in $P \oplus Q$.
\end{corollary}

\begin{proof}
  \Obda assume that $B_Q$ does not fit into $\mathring B_H$ with respect to the
  $x$-direction. Consider the two extreme points on $\partial Q$ in
  that direction and connect them by a closed path $\gamma$, which
  obviously does not fit into $-H$, as it does not fit into $\mathring B_H$.
\end{proof}

\noindent

\begin{theorem}\label{theorem2}
Let $P$ and $Q$ be two polygons with holes and let $P'$ and 
$Q'$ be their filtered versions, that is, with holes filled
up according to Corollary~\ref{cor:hole_filter} with $P\oplus Q = P'\oplus Q'$. 
Then, at least $P'$ or $Q'$ is a simple polygon.
\end{theorem}

\begin{proof}
Note that if $B_Q$ does not fit in the open axis-aligned bounding box
$\mathring B_P$ of $P$, it cannot fit in the bounding box of any hole
in~$P$, implying that all holes of $P$ can be ignored. Since for any
two bounding boxes either $B_Q \not \subset \mathring B_P$ or $B_P
\not \subset \mathring B_Q$ holds, we need to consider the holes of at
most one polygon.
\end{proof}

Consequently, we can remove all holes whose bounding boxes are, in
$x$- or $y$-direction, smaller than, or as large as, the bounding box
of the other polygon, as an initial phase of all methods. With fewer
holes, convex decomposition results in fewer pieces. Moreover, when
all holes of a polygon become irrelevant, one can choose a
decomposition method that handles only simple polygons instead of a
decomposition method that handles polygons with holes, which is
typically more costly. As for the convolution approach, the
intermediate arrangements become smaller, speeding up the algorithm.

\section{Implementation}
\label{sec:implementation}
The software has been developed as part of the
\iiDMinkowskiSumsPackage{} package of \cgal{}~\cite{cgal:w-rms2-15},
and it uses other components of \cgal~\cite{cgal:eb-15}.
As such, it is written in \C++ and rigorously adheres to the
generic-programming paradigm the EGC paradigms. In the following we
provide some details about each one of the new implementations.

\subsection{Reduced Convolution}
\label{ssec:impl:rc}
We compute the reduced convolution set of segments filtering out
features that cannot possibly contribute to the boundary of the
Minkowski sum (see Section~\ref{ssec:terms}) and in particular
complete holes (see Section~\ref{sec:filtration}). Then, we construct
the arrangement induced by the reduced convolution
set.\footnote{Currently, we use a single arrangement and do not
  separate segments that originate from different boundary cycles in the
  summands (exploiting Corollary~\ref{cor:one-trace}). We plan to apply this
  enhancement in the near future.} Finally, we traverse the
arrangement and extract the boundary of the Minkowski sum. We apply
two different filters to identify valid holes in the Minkowski sum:
(i) We ignore any face in the arrangement the outer boundary of which
forms a cycle that is not properly oriented, as suggested
in~\cite{bl-frmsu-11}.  (ii) We ignore any face $f$, such that $(-P
\oplus x)$ and $Q$ collide, where $x \in f$ is a sampled point inside
$f$, as suggested in~\cite{kos-cmsrp-91}. We use axis-aligned bounding
box trees to expedite the collision tests. After applying these two
filters, only segments that constitute the Minkowski sum boundary
remain.

\subsection{Decomposition}
\label{ssec:impl:dec}
Vertical decomposition~\cite{h-a-04} (a.k.a. trapezoidal
decomposition) and triangulation~\cite{b-tmg-04} have been extensively
used ever since they have been independently introduced a long time
ago. We provide a brief overview of these two structures for
completeness and explain how they are used in our implementations.

\begin{figure}[!thb]
  \centering
  \subfloat[][]{\label{fig:decomp:polygon}%
    \begin{tikzpicture}[scale=0.5]
    \filldraw[color=polyADarkColor,fill=polyALightColor,thick]%
      (0,0)--(6,0)--(6,6)--(0,6)--cycle;
    \filldraw[color=polyADarkColor,fill=white,thick]%
      (1,3)--(3,5)--(5,3)--(3,1)--cycle;
    \end{tikzpicture}}\quad
  \subfloat[][]{\label{fig:decomp:vd}%
    \begin{tikzpicture}[scale=0.5]
    \filldraw[color=polyBDarkColor,fill=polyBLightColor,thick]%
      (0,0)--(6,0)--(6,6)--(0,6)--cycle;
    \filldraw[color=polyBDarkColor,fill=white,thick]%
      (1,3)--(3,5)--(5,3)--(3,1)--cycle;
    \draw[black,thick,<->](1,0)--(1,6);
    \draw[black,thick,<-](3,0)--(3,1);
    \draw[black,thick,->](3,5)--(3,6);
    \draw[black,thick,<->](5,0)--(5,6);
    \end{tikzpicture}}\quad
  \subfloat[][]{\label{fig:decomp:tri}%
    \begin{tikzpicture}[scale=0.5]
    \filldraw[color=polyCDarkColor,fill=polyCLightColor,thick]%
      (0,0)--(6,0)--(6,6)--(0,6)--cycle;
    \filldraw[color=polyBDarkColor,fill=white,thick]%
      (1,3)--(3,5)--(5,3)--(3,1)--cycle;
    \draw[black,thick](0,0)--(1,3);
    \draw[black,thick](0,0)--(3,1);
    \draw[black,thick](6,0)--(5,3);
    \draw[black,thick](6,0)--(3,1);
    \draw[black,thick](6,6)--(5,3);
    \draw[black,thick](6,6)--(3,5);
    \draw[black,thick](0,6)--(1,3);
    \draw[black,thick](0,6)--(3,5);
  \end{tikzpicture}}
  \caption{\cpationstyle Convex decomposition.
    \protect\subref{fig:decomp:polygon} A polygon with holes.
    \protect\subref{fig:decomp:vd} Vertical decomposition of the polygon in \protect\subref{fig:decomp:polygon}.
    \protect\subref{fig:decomp:tri} Triangulation of the polygon in \protect\subref{fig:decomp:polygon}.}
  \label{fig:decomp}
\end{figure}

Vertical decomposition for a planar subdivisions is the partition of
the (already subdivided) plane into a finite collection of pseudo
trapezoids. Each pseudo trapezoid is either a trapezoid that has
vertical sides, or a triangle (which is a degenerate trapezoid). Given
a polygon with holes, we obtain the decomposition as follows: At every
vertex of the polygon, we extend a ray upward if it does not
escape the polygon, until either another vertex or an edge is
hit. Similarly, we extend a ray downward; see
Figure~\ref{fig:decomp:vd}. In our implementation we exploit the
vertical decomposition functionality provided by the \cgal{} package
\iiDArrangementsPackage{}~\cite{cgal:wfzh-a2-15}.

A Delaunay triangulation for a set of points in a plane is the
partition of the plane into triangles, such that no point in the input
is inside the circumcircle of any triangle in the triangulation.  A
constrained Delaunay triangulation is a generalization of the Delaunay
triangulation that forces certain required segments into the
triangulation. Given a polygon with holes we obtain the constrained
Delaunay triangulation confined to the given polygon and provide the
polygon edges as constraints; see Figure~\ref{fig:decomp:tri}.  In our
implementation, we use the
\iiDTriangulationsPackage{}~\cite{cgal:y-t-15} \cgal{} package.

\ignore{
\subsection{Performance}
In most cases, the reduced convolution method consumes less time than
the full convolution method, as the induced arrangement has typically
fewer cells. 
However, in degenerate cases with many holes in the
Minkowski sum, the full convolution method is preferable over the
reduced convolution method, as it avoids the costly
collision-detection tests. 
The decomposition methods consume more time
than the convolution methods in almost all cases. During the
application of the vertical decomposition new vertices are
constructed. The coordinates of the newly constructed vertices tend to
have large bit lengths. This could be devastating when the EGC
paradigm is in use. Triangulation could yield a number of triangles
that is much larger than the number of pseudo trapezoids of the
vertical decomposition. A large number of convex pieces increases the
pressure on the subsequent union step of the decomposition
method. While the time consumption of the union step heavily depends
on the output of the proceeding convex decomposition step, the two steps
are decoupled. Any optimization applied to the union step can
significantly improve the overall performance.
}

\section{Experiments}
\label{sec:experiments}
We have conducted our experiments on families of randomly generated
simple and general polygons from
AGPLib~\cite{art-gallery-instances-page}; examples are depicted in
Figure~\ref{fig:rand_simple} and~\ref{fig:rand_nonsimple},
respectively.  All experiments were run on an \emph{Intel Core 2 Duo
  P9600} CPU clocked at~2.53~GHz with~4~GB of RAM.  For each instance
size the diagrams in the figures show an average over~10 runs on
different input. Every run was allowed 20 minutes of CPU time 
and aborted when it did not finish within this limit.

\begin{figure}[thb]
  \centering
  \subfloat[][]{\label{fig:rand_simple}%
    \includegraphics[width=3cm,height=3cm]{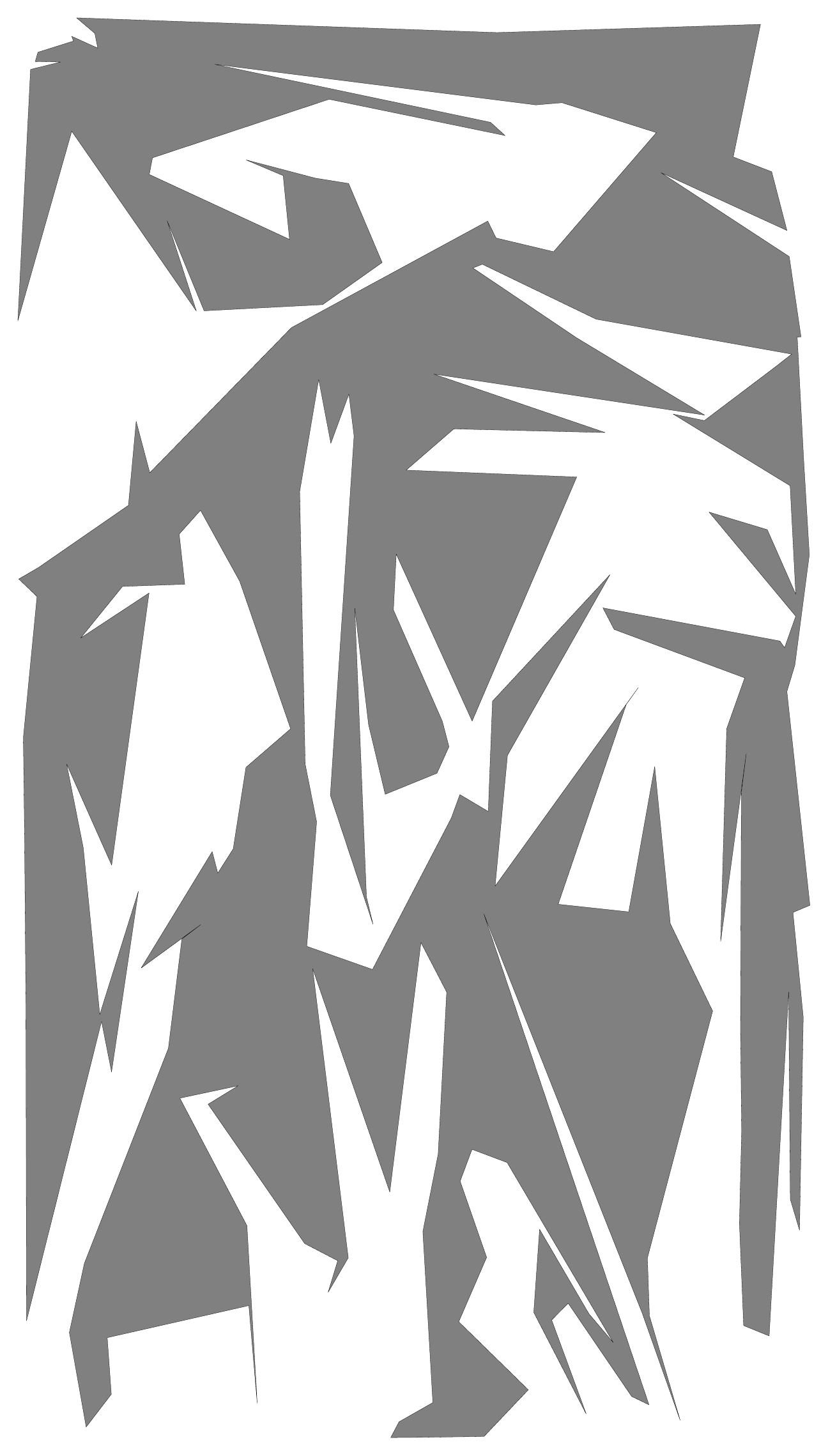}}\hfil
  \subfloat[][]{\label{fig:rand_nonsimple}%
  \includegraphics[width=3cm,height=3cm]{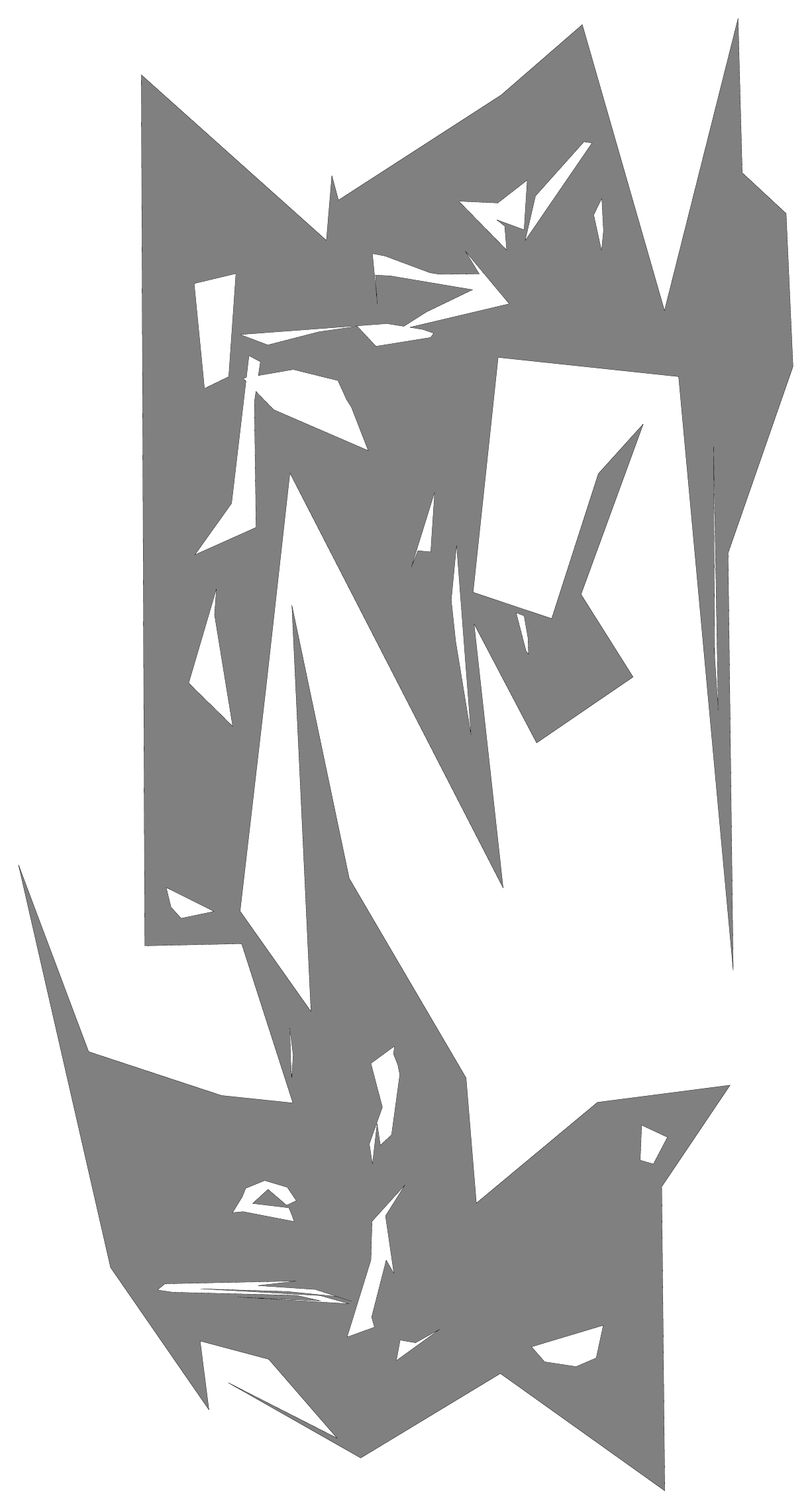}}
  \caption{\cpationstyle Randomly generated polygons: (a) simple polygon with 200 vertices, and
  (b) general with 200 vertices and 20 holes.}
\end{figure}

First, we compared the running time of the implementations of all
methods for simple polygons available in \cgal{} (for details, see
\cite[Section~9.1.2]{fhw-caass-12}), the new
implementations, and Behar and Lien's implementation; see
Figure~\ref{fig:simple_simple}.  
The reduced convolution method consumed about
ten times less time than the full convolution method for large
instances, whereas the decomposition methods were the fastest for
instances larger than~150 vertices. \todo[inline]{why?}

Secondly, we compared the running time of the implementations of the
three new methods (\ie~the reduced convolution (RC), the triangular
decomposition (TD), and the vertical decomposition (VD)) and Behar and Lien's
implementation on instances of general polygons with $n$ vertices and
$n$/10 holes; see Figure~\ref{fig:pwh_pwh}. For each pair of polygons,
one was scaled down by a factor of~1000, to avoid the effect of the
hole filter in this experiment.
For all executions, the reduced convolution method
consumed significantly less time than the two decomposition
methods. Behar and Lien's implementation generally performs worse than
our reduced convolution method.

In order to demonstrate the effect of the hole filter, 
we compared the running time of the implementations above fed
with a square of varying size (see the horizontal axis in
Figure~\ref{fig:growing_square-no-filter} and
~\ref{fig:growing_square}) and with randomly generated
polygons having~2000 vertices and 200~holes.  
Without the hole filter the running time of the reduced convolution 
method increases as the square grows due to
an increase of the complexity of the intermediate arrangement. 
Behar and Lien's implementation exhibited constant running time, 
as it performs pairwise intersection testing.
When applying the hole filter to our methods,
the reduced convolution method consumed less time than all other methods.  
The two diagrams clearly show the impact of filtering holes.


\begin{figure}[thb]
    \centering
    \subfloat[][]{\label{fig:letter_m}%
    \includegraphics[height=2.5cm]{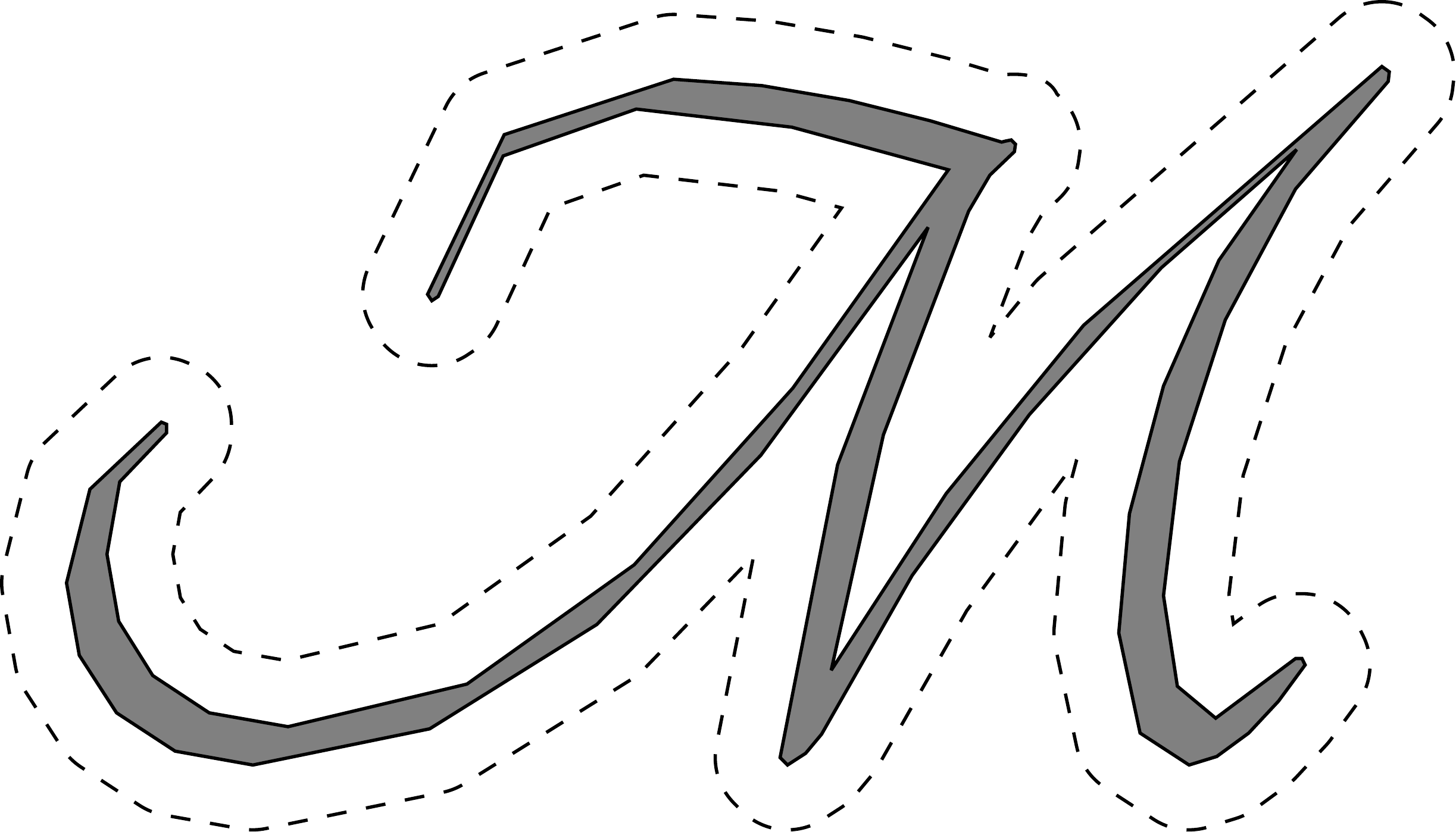}}\hfil
    \subfloat[][]{\label{fig:letter_a}%
    \includegraphics[height=2.5cm]{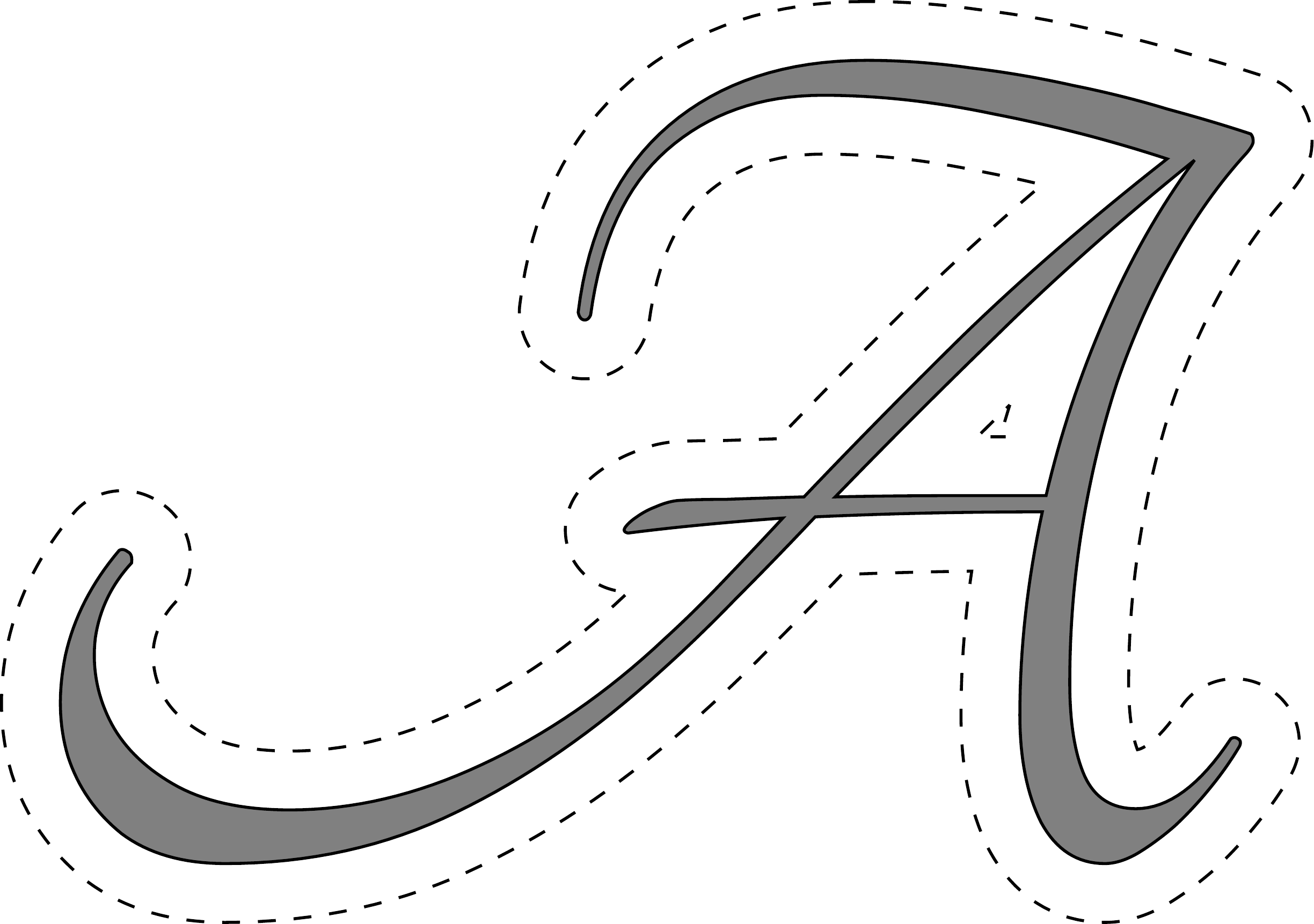}}
    \caption{\cpationstyle Letters from the font \emph{Tangerine} used for the real-world benchmark, displayed with their offset versions. (a) Lowest-resolution “M” with 75 vertices (b) Highest-resolution “A” with 8319 vertices. \label{fig:letter}}
\end{figure}

\begin{figure}[htb]
\center
  \subfloat[][MS of simple polygons;]{\label{fig:simple_simple}
    \includegraphics[width=0.43\columnwidth]{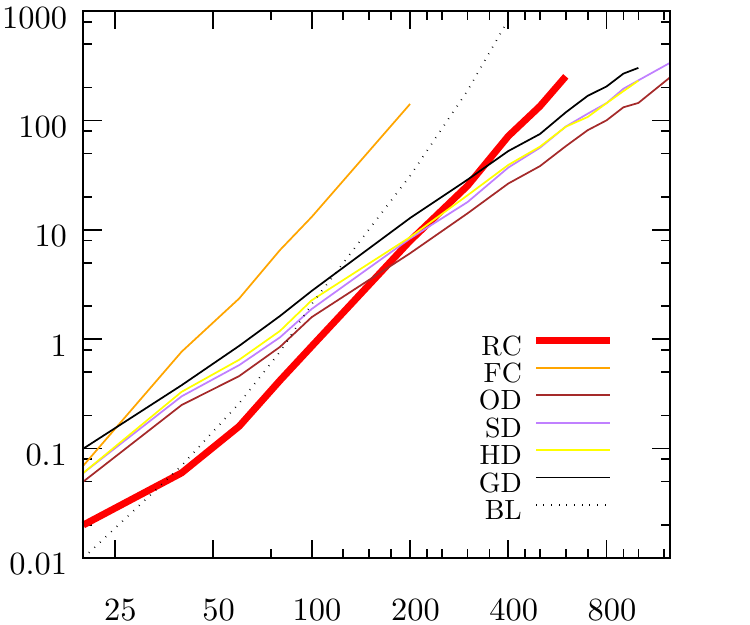}}
  \hspace{0.1cm}
  \subfloat[][MS of general polygons.]{\label{fig:pwh_pwh}
    \includegraphics[width=0.43\columnwidth]{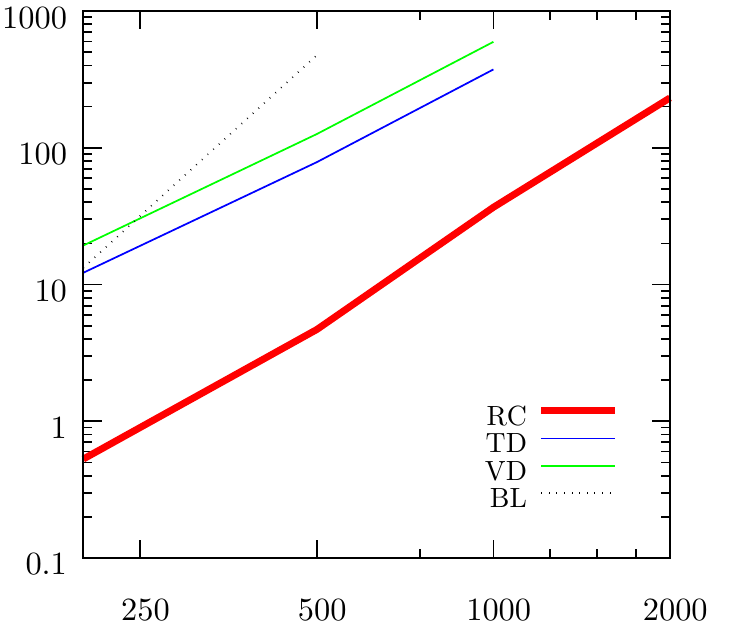}}\\
  \subfloat[][MS of general polygon (200 holes, 2000 vertices) and growing square ($x$-axis)---without hole filter]{\label{fig:growing_square-no-filter}
    \includegraphics[width=0.43\columnwidth]{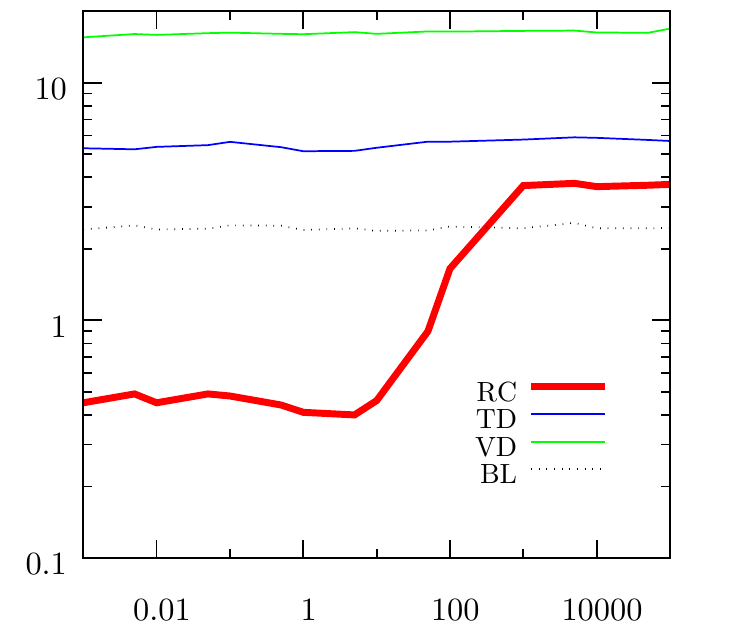}}
  \hspace{0.1cm}
  \subfloat[][MS of general polygon (200 holes, 2000 vertices) and growing square ($x$-axis)---with hole filter]{\label{fig:growing_square}
    \includegraphics[width=0.43\columnwidth]{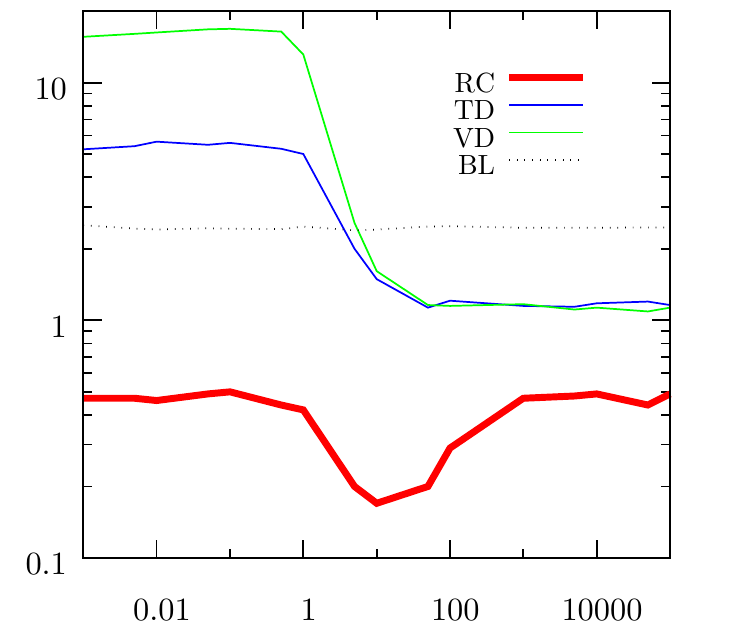}}\\
    \subfloat[][MS of a fixed-size circle and an “M” with varying vertex count ($x$-axis).]{\label{fig:m}
    \includegraphics[width=0.43\columnwidth]{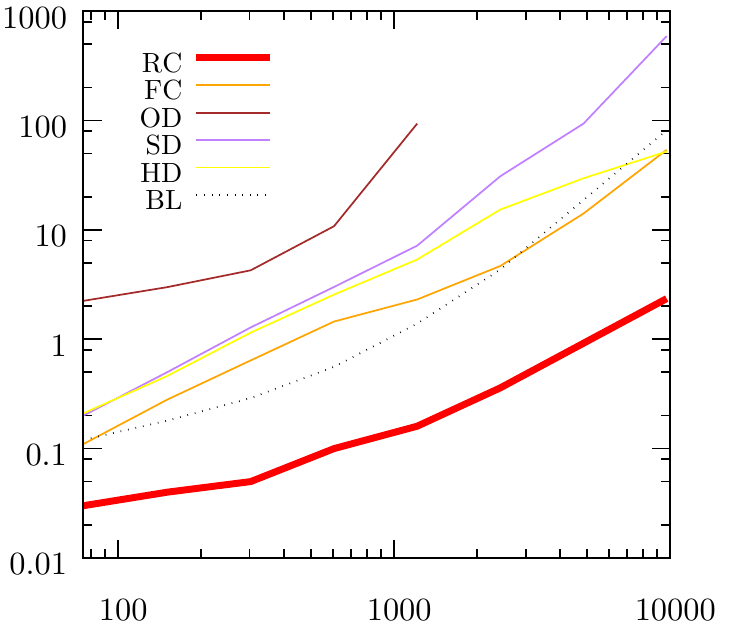}}
  \hspace{0.1cm}
  \subfloat[][MS of a fixed-size circle and an “A” with varying vertex count ($x$-axis).]{\label{fig:a}
    \includegraphics[width=0.43\columnwidth]{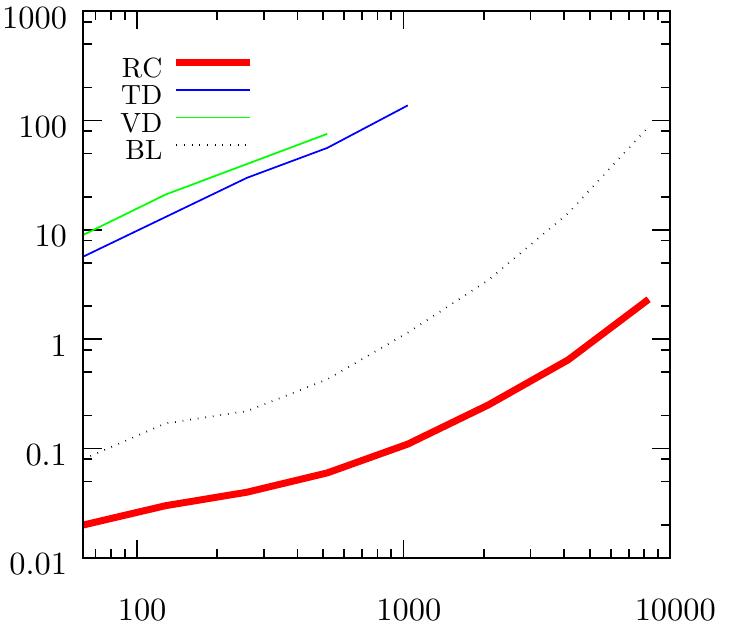}}\\
  \caption{\cpationstyle  
    Time in seconds for different methods to compute Minkowski sums for two polygons.  
    $x$-axis: \#vertices of each input polygon, if not stated otherwise.
    {\footnotesize Legend:
      (RC) reduced convolution;
      (FC) full convolution;
      (TD) constrained triangulation decomposition;
      (VD) vertical decomposition;
      (SD) small-side angle-bisector decomposition;
      (OD) optimal convex decomposition;
      (HD) Hertel-Mehlhorn decomposition;
      (GD) Greene decomposition;
      (BL) Behar and Lien's reduced convolution.}}
      \label{fig:plots}
\end{figure}

\FloatBarrier
Note that the polygons used for the benchmarks above do not represent
a real-world case.  Instead, the complex shapes rather constitute a
worst-case, as most segments intersections are inside the Minkowski
sum anyway.  For a more realistic scenario, consider a text, which we
want to offset (for example, for printing stickers). In
Figure~\ref{fig:m}, we show the running times of the methods available
for simple polygons when calculating the Minkowski sum of a letter “M”
(Figure~\ref{fig:letter}) with a varying amount of vertices and a
circle with 128 vertices.  In Figure~\ref{fig:a}, we show the running
times of the methods available for general polygons when calculating
the Minkowski sum of a letter “A” and the same circle. 
For both letters, our implementation of the reduced convolution 
is at least 5 times faster than all other methods.

\section{Conclusion}
\label{sec:conclusion}

All new implementations introduced in this work 
will be available as part of the next public release of
CGAL, which now also supports polygons with holes. 
The  decomposition approaches that handle only simple 
polygons outperform the new reduced convolution method 
(which, naturally handles also simple polygons) 
for instances of random simple polygons with more than 150 vertices.
However, these rather chaotic polygons somewhat constitute 
the worst case scenario for the reduced convolution method. 
In all other scenarios, the reduced convolution method 
with hole filter outperforms all other methods by a 
factor of at least 5. 
Consequently, this is the new default method of CGAL to 
compute Minkowski sums for simple polygons as well as 
polygons with holes.


{\small
\bibliographystyle{abbrv}
\bibliography{abrev-short,bibliography}}

\end{document}